\documentclass[%
 reprint,
superscriptaddress,
 amsmath,amssymb,
 aps,
prl,
]{revtex4-1}

\usepackage{amsthm}
\usepackage{graphicx}
\usepackage{dcolumn}
\usepackage{bm}
\usepackage{multirow} 

\DeclareMathOperator{\supp}{supp}
\DeclareMathOperator{\conf}{conf}

\theoremstyle{remark}

\newtheorem{theorem}{\indent \emph{\textbf{Theorem}}}

\begin{document}

\makeatletter
\newcommand{\rmnum}[1]{\romannumeral #1}
\newcommand{\Rmnum}[1]{\expandafter\@slowromancap\romannumeral #1@}
\makeatother

\preprint{APS/123-QED}

\title{Quantum algorithm for association rules mining}

\author{Chao-Hua Yu}
\affiliation{State Key Laboratory of Networking and Switching Technology, Beijing University of Posts and Telecommunications, Beijing, 100876, China}
\affiliation{State Key Laboratory of Cryptology, P.O. Box 5159, Beijing, 100878, China}
\author{Fei Gao}
\email{gaof@bupt.edu.cn}
\affiliation{State Key Laboratory of Networking and Switching Technology, Beijing University of Posts and Telecommunications, Beijing, 100876, China}
\author{Qing-Le Wang}%
\affiliation{State Key Laboratory of Networking and Switching Technology, Beijing University of Posts and Telecommunications, Beijing, 100876, China}
\author{Qiao-Yan Wen}
\affiliation{State Key Laboratory of Networking and Switching Technology, Beijing University of Posts and Telecommunications, Beijing, 100876, China}

\date{\today}

\begin{abstract}
Association rules mining (ARM) is one of the most important problems in knowledge discovery and data mining. The goal of it is to acquire consumption habits of customers by discovering the relationships between items from a transaction database that has a large number of transactions and items. In this paper, we address ARM in the quantum settings and propose a quantum algorithm for the most compute intensive process in ARM, i.e., finding out the frequent 1-itemsets and 2-itemsets. In our algorithm, to mine the frequent 1-itemsets efficiently, we use the technique of amplitude amplification. To mine the frequent 2-itemsets efficiently, we introduce a new quantum state tomography scheme, i.e., pure-state-based tomography. It is shown that our algorithm is potential to offer polynomial speedup over the classical algorithm.
\begin{description}
\item[PACS numbers]
03.67.Dd, 03.67.Hk
\end{description}
\end{abstract}

\pacs{Valid PACS appear here}
\maketitle



\emph{Introduction.---}Quantum computing provides a paradigm that makes use of quantum mechanical principles, such as superposition and entanglement, to perform computing tasks in quantum systems (quantum computers) \cite{QCQI}. Just as classical algorithms run in the classical computers, a quantum algorithm is a step-by-step procedure run in the quantum computers for solving a certain problem, which, more interestingly, is expected to outperform the classical algorithms for the same problem. As of now, various quantum algorithms have been put forward to solve a number of problems faster than their classical counterparts \cite{QA}, and mainly fall into one of three classes \cite{PWShor}. The first class features the famous Shor's algorithm \cite{Shor} for large number factoring and discrete logarithm, which offers exponential speedup over the classical algorithms for the same problems. The second class are represented by the Grover's quantum search \cite{Grover} and its generalized version, i.e., amplitude amplification \cite{AA}, which achieve quadratic speedup over the classical search algorithm. The third class contains the algorithms for quantum simulation \cite{QS}, the original idea of which is suggested by Feynman \cite{F} to speed up the simulation of quantum systems using quantum computers.

In the past decade, quantum simulation has made great progress in efficient sparse Hamiltonian simulation \cite{BCK}, which underlies two important quantum algorithms, quantum algorithm for solving linear equations (called HHL algorithm) \cite{HHL} and quantum principal component analysis \cite{QPCA}. The former is to generate a pure quantum state encoding the solution of linear equations, which is potential to achieve exponential speedup over the best classical algorithm for the same problem. The latter is an efficient quantum state tomography on quantum sates with low-rank or approximately low-rank density matrix based on the technique of density matrix exponentiation. Inspired by these two algorithms, a number of quantum machine learning algorithms for big data have been proposed and potentially exhibit exponential speedup over the classical algorithms \cite{WBL,LMR,RML,CD}. For example, the quantum support vector machine was recently proposed for big data classification \cite{RML}. These quantum machine learning algorithms will evidently make the tasks of big data mining be accomplished more efficiently than their classical counterparts.

In this letter, we address another important problem in big data mining, association rules mining (ARM), in the quantum settings. The goal of ARM is to acquire consumption habits by mining association rules from a large transaction database \cite{DM,Apriori,PCY}. More formally, given a transaction database consisting of a large number of transactions and items, the task of ARM is to discover the association rules connecting two itemsets (an itemset is a set of items) $A$ and $B$ in the conditional implication form $A \Rightarrow B$, which implies that a customer who buys the items in $A$ also tends to buy the items in $B$. The core of ARM is to mine the itemsets that frequently occur in the transactions. In the classical regime, various algorithms for mining frequent itemsets have been proposed and well studied over the past decades \cite{DM}, the most famous one being the Apriori algorithm \cite{Apriori}. However, the information explosion today makes the database to be processed extremely large, which poses great challenges to the compute ability of classical computers for undertaking ARM by using these algorithms. Therefore, it is of great significance to propose more efficient algorithms for ARM.

In practice, the processing cost of mining frequent 1-itemsets ($k$-itemset is a set of $k$ items) and 2-itemsets dominates the total processing cost of mining all the frequent itemsets in ARM \cite{PCY}. Therefore, in this letter, we propose a quantum algorithm for mining frequent 1-itemsets and 2-itemsets. In our algorithm, the technique of amplitude amplification \cite{AA} is applied to efficiently mine frequent 1-itemsets. To efficiently mine frequent 2-itemsets, we introduce a new quantum state tomography scheme, i.e., pure-state-based tomography, which could be applied to efficiently reconstruct the approximately low-rank density matrix with nonnegative elements. It will be shown that this algorithm is potential to achieve polynomial speedup over the classical sampling algorithm.



\emph{Review of ARM.---}We first review some basic concepts and notations in ARM. Suppose a transaction database, the objective ARM deals with, contains $N$ transactions $T=\{T_1,T_2,\cdots,T_N \}$ and each one is a subset of $M$ items $I=\{I_1,I_2,\cdots,I_M \}$, i.e., $T_i \subseteq I$. It can also be seen as $N \times M$ binary matrix denoted by $D$ in which the element $D_{ij}=1 (0)$ means that the item $I_j$ is (not) contained in the transaction $T_i$.
The \emph{support} of an itemset $X$ is defined as the proportion of transactions in $T$ that contain all the items in $X$, i.e., $\supp(X)=\frac{|\{T_i| X \subseteq T_i\}|}{N}$. An association rule connects two disjoint itemsets $A$ and $B$ and is of the implication form $A \Rightarrow B$.
Its support (confidence) is defined as $\supp(A \Rightarrow B)=\supp(A\cup B)$ ($\conf(A \Rightarrow B)=\frac{\supp(A\cup B)}{\supp(B)}$). A rule is called frequent (confident) if its support (confidence) is not less than a prespecified threshold $min\_supp$ ($min\_conf$). The task of ARM is to find out the rules $A \Rightarrow B$ that are both frequent and confident. This can be achieved by two phases: (1) find out all the frequent itemsets $X$, defined as $\supp(X)>min\_supp$; (2)find out all the rules $A \Rightarrow B$ such that $A\cup B=X$ and it is confident. Because the second phase is much less costly than the first, the core work of ARM lies in the first phase. Furthermore, the processing cost in discovering frequent 1-itemsets and 2-itemsets, denoted by $F_1$ and $F_2$ respectively, dominates the total processing cost of the first phase. Therefore, how to reduce the cost of discovering $F_1$ and $F_2$ is of great significance.

Based on the Apriori property stating that all nonempty subset of a frequent itemset must also be frequent, mining $F_1$ and $F_2$ can be done in a level-wise manner. Firstly, one can compute all the supports of each item and pick out the frequent items to constitute $F_1$. Secondly, one can use $F_1$ to generate the candidate 2-itemsets $C_2=F_1 \bowtie F_1=\{\{I_i,I_j\}|I_i,I_j \in F_1, I_i\neq I_j\}$, compute the supports of itemsets in $C_2$ and then pick out the frequent ones to constitute $F_2$. In fact, computing the support of any 1-itemset or 2-itemset can be transformed into computing the inner product of two binary vectors. Suppose the transaction database corresponds to a $N\times M$ binary matrix $D=(\overrightarrow{d_1},\overrightarrow{d_2},\cdots,\overrightarrow{d_M})$, then the support of any 1-itemset $I_i$ can be computed by $S_{ii}=\frac{\overrightarrow{d_i}\cdot\overrightarrow{d_i}}{N}$ and the support of any 2-itemset $\{I_i,I_j\}$ can be computed by $S_{ij}=\frac{\overrightarrow{d_i}\cdot\overrightarrow{d_j}}{N}$. Therefore, the supports of all the 1-itemsets and 2-itemsets can be computed by $S=\frac{D^TD}{N}$, where $D^T$ is the transpose of the matrix $D$ and the supports of all the 1-itemsets (2-itemsets) correspond to the diagonal (off-diagonal) elements of $S$.

\emph{Pure-state-based quantum state tomography.---}
Before giving the details of our quantum algorithm for mining $F_1$ and $F_2$, we first introduce a new quantum state tomography, pure-state-base tomography, which will be used to mine $F_2$ as a key subroutine. Quantum state tomography is a process of reconstructing the density matrix of an unknown quantum state by performing series of measurements on a large number of copies of this state. In general, tomography on a $d$-dimensional quantum state requires $d^2$ measurement settings \cite{GLM}. However, in many cases, the density matrix of the state could be of low or approximately low rank \cite{QPCA,GLFBE}. That is, it can be approximately constructed by $r\ll d$ largest eigenvalues and their corresponding eigenvectors from the spectral decomposition. In this case, two recently invented tomography techniques, quantum state tomography via compressed sensing \cite{GLM} and quantum principal component analysis \cite{QPCA}, can be applied. Both schemes require only $\mathcal{O}(rdpoly(\log(d)))$ measurement settings, which offer significant improvement on large quantum systems.

However, all the above tomography schemes will perform postprocessing in the classical computer on the measurement outcomes to reconstruct the classical description of the density matrix,  and this will take time $\Omega(d^2)$ because $d^2$ elements of the density matrix need to be determined. Here we propose a new quantum state tomography scheme, named \emph{pure-state-based quantum state tomography}, that is potential to overcome this limit and more direct to obtain the elements of density matrix. In our scheme, we do not directly perform tomography on the state written as $\rho=\sum_{i,j=1}^d\rho_{ij}|i\rangle\langle j|$ but transform it into a pure state approximating $\frac{\sum_{i,j=1}^d\rho_{ij}|i\rangle|j\rangle}{\sqrt{\sum_{i,j=1}^d|\rho_{ij}|^2}}$. Once the pure state is created, one can perform measurements on the pure state to reveal the information of $\rho_{ij}$ with bounded error. To illustrate it, we first show how to prepare this pure state in the following.

\begin{theorem}
Suppose the a $d$-dimensional quantum state with density matrix $\rho=\sum_{i,j=1}^d\rho_{ij}|i\rangle\langle j|$ has eigenvalues $\lambda_j$ satisfying $\frac{1}{\kappa}\leq\lambda_j\leq1$ and can be generated in time $T_{\rho}$, then a pure state $|\psi\rangle$ approximating $|\psi_{\rho}\rangle=\frac{\sum_{i,j=1}^d\rho_{ij}|i\rangle|j\rangle}{\sqrt{\sum_{i,j=1}^d|\rho_{ij}|^2}}$ as $\||\psi\rangle-|\psi_{\rho}\rangle\|\leq \epsilon$ can be created taking $\mathcal{O}(\frac{\kappa^3}{\epsilon^3})$ copies of $\rho$ and time $\mathcal{O}(\frac{T_{\rho}\kappa^3}{\epsilon^3})$.
\end{theorem}

\begin{proof}
The idea for creating the pure state $|\psi\rangle$ is that performing the operation $\rho \otimes I$ on the state vector $\frac{\sum_{k=1}^{d}|k\rangle|k\rangle}{\sqrt{d}}$ will yield
\begin{eqnarray}
\frac{\sum_{k=1}^{d}\{(\sum_{i,j=1}^{d}\rho_{ij}|i\rangle\langle j|)|k\rangle\}|k\rangle}{\sqrt{d}}=\frac{\sum_{i,k=1}^{d}\rho_{ik}|i\rangle |k\rangle}{\sqrt{d}}.
\label{eq:1}
\end{eqnarray}
Performing the operation can be achieved by the technique of improved phase estimation together with controlled rotation operation which has been applied in HHL algorithm \cite{HHL} and a number of quantum machine learning algorithms \cite{WBL,LMR,RML,CD}. The detailed steps are presented as follows:

1. Prepare three quantum registers in the initial state
\begin{eqnarray}
(\sqrt{\frac{2}{t}}\sum_{\tau=0}^{t-1}\sin(\frac{\pi(\tau+1/2)}{t})|\tau\rangle)(\frac{\sum_{k=1}^{d}|k\rangle|k\rangle}{\sqrt{d}}).
\label{eq:2}
\end{eqnarray}

2. Perform the controlled unitary operation $\sum_{\tau=0}^{t-1}|\tau\rangle\langle\tau|\otimes e^{\frac{-i\rho \tau t_0}{t}}$ on the first two registers. Here $t_0$ and $t$ are two parameters introduced to adjust the accuracy. Implementing the operation requires the techniques of density matrix exponentiation and its controlled fashion\cite{QPCA}, which takes $t$ copies of $\rho$ and introduces error $\mathcal{O}(\frac{t_0^2}{t})$\cite{QPCA,RML}. Thus, to ensure the error is within $\epsilon$, $t$ is set $t=\mathcal{O}(\frac{t_0^2}{\epsilon})$.

3. Apply Fourier transformation on the first register to obtain a state close to $\frac{\sum_{j,k=1}^d|\lambda_j\rangle \langle u_j|k\rangle |u_j\rangle|k\rangle}{\sqrt{d}}$, where $\lambda_j$ and $|u_j\rangle$ are the eigenvalues and eigenvectors of $\rho$, i.e., $\rho=\sum_{j=1}^d\lambda_j |u_j\rangle\langle u_j|$. The error of eigenvalue in the first register is within $\mathcal{O}(\frac{1}{t_0})$ and the error induced of final state is within $\mathcal{O}(\frac{\kappa}{t_0})$ \cite{HHL}. Therefore, in order to ensure the final error is within $\epsilon$, $t_0$ should be set $t_0=\mathcal{O}(\frac{\kappa}{\epsilon})$ and $t=\mathcal{O}(\frac{\kappa^2}{\epsilon^3})$.

4. Introduce an auxiliary qubit in the state $|0\rangle$ and perform the controlled unitary operation on the eigenvalue register and the auxiliary qubit to generate the state approximating
$\frac{\sum_{j,k=1}^d|\lambda_j\rangle \langle u_j|k\rangle |u_j\rangle|k\rangle(\sqrt{1-|C\lambda_j|^2}|0\rangle+C\lambda_j|1\rangle)}{\sqrt{d}}$,
where $C\in\mathcal{O}(max_{j}\lambda_j)^{-1}$.

5. Erase the eigenvalue register and the state is near to $\frac{\sum_{j,k=1}^d\langle u_j|k\rangle |u_j\rangle|k\rangle(\sqrt{1-|C\lambda_j|^2}|0\rangle+C\lambda_j|1\rangle)}{\sqrt{d}}$.

6. Measure the last qubit until obtaining the outcome $|1\rangle$. Then we obtain the state $|\psi\rangle$ close to
\begin{eqnarray}
\frac{\sum_{j,k=1}^d\lambda_j\langle u_j|k\rangle |u_j\rangle|k\rangle}{B}&=&\frac{(\rho \otimes I)\sum_{k=1}^{d}|k\rangle|k\rangle}{B}\nonumber \\
&=& \frac{\sum_{i,k=1}^{d}\rho_{ik}|i\rangle |k\rangle}{B},
\label{eq:3}
\end{eqnarray}
i.e., $|\psi_{\rho}\rangle$, where $B=\sqrt{\sum_{l=1}^d \lambda_l^2}=\sqrt{\sum_{i,j=1}^d|\rho_{ij}|^2}$. Obviously, this state vector is proportional to the vector (1). The probability of obtaining $|1\rangle$ is $\sum_{j=1}^{d}\frac{|C\lambda_j|^2}{d} \in \Omega(\frac{1}{\kappa^2})$. Thus $\mathcal{O}(\kappa^2)$ measurements are needed to obtain this outcome with a large probability. The technique of amplitude amplification \cite{AA} can be applied to reduce the number of repetitions to $\mathcal{O}(\kappa)$.

According to the steps 3 and 6, we can see that the number of copies of $\rho$ required to create $|\psi\rangle$ scales as $\mathcal{O}(t\kappa)=\mathcal{O}(\frac{\kappa^3}{\epsilon^3})$ and thus the time complexity scales as $\mathcal{O}(\frac{T_{\rho}\kappa^3}{\epsilon^3})$. Here the time for simulating the SWAP operator (in step 2) for time $\frac{t_0}{t}$ is too small to be neglected \cite{RML}.
\end{proof}

However, for the state $\rho$ of low or approximately low rank, the lower bound $\frac{1}{\kappa}$ will be extremely small, or equivalently, $\kappa$ will be extremely large. According to the theorem above, it will make creating the pure state $|\psi\rangle$ very expensive in time. To overcome this problem, inspired by \cite{HHL,RML}, a constant $\epsilon_{eff} = \mathcal{O}(1)$ is chosen and only the eigenvalues $\epsilon_{eff} \leq\lambda_j\leq1$ are taken into considered. In this case, in step 6, the probability of obtaining $|1\rangle$ will be $\Omega(\frac{\epsilon_{eff}^2}{d})$. Consequently, it will take $\mathcal{O}(\frac{\sqrt{d}}{\epsilon_{eff}^3\epsilon^3})$ copies of $\rho$ and time $\mathcal{O}(\frac{T_{\rho}\sqrt{d}}{\epsilon_{eff}^3\epsilon^3})$ to create $|\psi\rangle$.

After creating the pure state $|\psi\rangle$, one can perform measurements on the state to reveal the estimates of $\rho_{ij}$. Take the case that $\rho$ is approximately low-rank and all of $\rho_{ij}$ are nonnegative real numbers as an example. One can perform the measurements in the computational basis on $|\psi\rangle$ and get the outcome $|i\rangle|j\rangle$ with proportion $p_{ij}$, then $\rho_{ij}$ can be estimated by $B\sqrt{p_{ij}}$. Here $B=\sqrt{\sum_{l=1}^d \lambda_l^2}$ can be estimated by estimating $\lambda_l$ via quantum principal component analysis \cite{QPCA}. This takes a little cost because $\rho$ is approximately low-rank. Moreover, assuming that there are $d'$ significantly large elements in $\rho$, or equivalently in $|\psi\rangle$, $\mathcal{O}(\frac{d'}{\epsilon^2})$ measurements in the computational basis are needed to approximate $|\psi\rangle$ and ensure sum of the squared errors is within $\epsilon^2$. The total runtime of tomography is $\mathcal{O}(\frac{T_{\rho}d'\sqrt{d}}{\epsilon_{eff}^3\epsilon^5})$.

\emph{Quantum Algorithm.---}

\emph{1. Mining frequent 1-itemsets.}
Mining frequent 1-itemsets requires estimating the supports of all 1-itemsets and determining the frequent 1-itemsets.
To achieve that in the quantum settings, we assume that the oracle accessing each element of the database binary matrix $D$ is provided. More precisely, it is a unitary operator $U_O$ acting on the computational basis,
\begin{eqnarray}
|i\rangle|j\rangle|0\rangle \xrightarrow{U_O} |i\rangle|j\rangle|D_{ij}\rangle.
\end{eqnarray}
This oracle can be implemented via quantum random access memory that takes time $\mathcal{O}(\log(NM))$ \cite{GLM}.

Provided with the above oracle, one can use the technique of amplitude amplification \cite{AA} to create a quantum state whose density matrix is proportional to the support matrix $S$ so that measuring the state in the computational basis can reveal the supports of all the 1-itemsets, i.e., $S_{ii}$. First, one prepare three quantum registers in the state $|\varphi_1\rangle=\frac{(\sum_{i=1}^{N}|i\rangle)\otimes(\sum_{j=1}^{M}|j\rangle)\otimes |0\rangle}{\sqrt{NM}}$. Secondly, perform $U_O$ on $|\varphi_1\rangle$ to yield the state $|\varphi_2\rangle=\frac{\sum_{i=1}^{N}\sum_{j=1}^{M}|i\rangle|j\rangle|D_{ij}\rangle}{\sqrt{NM}}$. Thirdly, we apply amplitude amplification to search the items $D_{ij}=1$ in the last qubit and then measure it until getting $|1\rangle$ to obtain the state $|\varphi_3\rangle=\frac{\sum_{i=1}^{N}\sum_{j=1}^{M}D_{ij}|i\rangle|j\rangle}{\sqrt{W}}$, where $W$ is the number of overall 1's in $D$, i.e., $W=\sum_{i=1}^M\sum_{j=1}^M D_{ij}$. The oracle complexity scales as $\mathcal{O}(\sqrt{\frac{NM}{W}})=\mathcal{O}(\sqrt{\frac{M}{a}})$ and the time complexity scales as $\mathcal{O}(\sqrt{\frac{M}{a}}\log(MN))$, where $a=\frac{W}{N}$ is the average number of items in each transaction and can be estimated by quantum counting with oracle complexity $\mathcal{O}(\sqrt{M})$ \cite{QC}. In practice, since a customer generally only buy few items in one transaction, $a$ generally scales as $\mathcal{O}(1)$. As a consequence, the oracle complexity comes to $\mathcal{O}(\sqrt{M})$. The density matrix of the second register of $|\varphi_3\rangle$ is $\rho=\frac{D^TD}{tr(D^TD)}=\frac{D^TD}{W}=\frac{S}{a}$, thus $S=a\rho$. Finally, measure the state $\rho$ in the computational basis and the outcome $i$ occurs with probability $\rho_{ii}=\frac{S_{ii}}{a}$, thus the supports $S_{ii}$ are estimated.

To obtain $F_1$, the classical Apriori algorithm determinately computes the supports of all the items which takes time $MN$. These supports can also be estimated via sampling technique. Take $N_c$ samples for estimating each support so that $S_{ii}$ can be estimated with squared error $\mathcal{O}(\frac{S_{ii}(1-S_{ii})}{N_c})$ and thus the sum of the squared errors for estimating all the supports scale as $\mathcal{O}(\frac{a}{N_c})$. To make the sum of the squared errors be less than a preset value $\epsilon^2$, $N_c$ is chosen as $\mathcal{O}(\frac{a}{\epsilon^2})=\mathcal{O}(\frac{1}{\epsilon^2})$ and thus the total time taken to estimate all the $M$ supports via sampling scales as $\mathcal{O}(\frac{M}{\epsilon^2})$. In our algorithm, we sample $\rho$ by measuring it in the computational basis. Suppose we take $N_q$ samples, $\rho_{ii}$ can be estimated with square error $\mathcal{O}(\frac{\rho_{ii}(1-\rho_{ii})}{N_q})$ and the sum of the squared errors scales as $\mathcal{O}(\frac{1}{N_q})$. Since $\rho_{ii}=\frac{S_{ii}}{a}$, the sum of the squared errors for estimating  $\rho_{ii}$ should be $a^2$ times that for estimating $S_{ii}$. Therefore, $N_q$ should be chosen $N_q=\mathcal{O}(\frac{a^2}{\epsilon^2})=\mathcal{O}(\frac{1}{\epsilon^2})$ and thus our algorithm takes time $\mathcal{O}(\frac{\sqrt{M}\log(MN)}{\epsilon^2})$ to mine frequent 1-itemsets.

\emph{2. Mining frequent 2-itemsets.}
Suppose there are $M_1$ frequent 1-itemsets, $F_1=\{I_{f_1},I_{f_2},\cdots,I_{f_{M_1}}\}$, based on which we can mine the frequent 2-itemsets $F_2$. To do that, we need to compute all the supports of the candidate 2-itemsets $C_2=F_1 \bowtie F_1=\{\{I_{f_i},I_{f_j}\}|I_{f_i},I_{f_j} \in F_1, I_{f_i} \neq I_{f_j}\}$ and pick out the frequent ones whose supports are not less than $min\_sup$. Just as the matrix $S$, all the supports of itemsets in $C_2$ are placed in the upper (or lower) off-diagonal elements of the matrix $\overline{S}=\frac{D_f^{T}D_f}{N}$, where $D_f=(\overrightarrow{d_{f_1}},\overrightarrow{d_{f_2}}, \cdots, \overrightarrow{d_{f_{M_1}}})$ is a part of $D$. Following the ideas of mining frequent 1-itemsets and pure-state-based quantum state tomography, our quantum algorithm proceeds as follows to mine frequent 2-itemsets:

1. Using the method in mining frequent 1-itemsets, a quantum state with density matrix $\sigma=\frac{D_f^TD_f}{tr(D_f^TD_f)}\propto \overline{S}$ is created taking time $\mathcal{O}(\sqrt{M_1}\log(M_1N))$. Note that the average number of items in each row of $D_f$, denoted by $a_f$, can also be estimated by quantum counting \cite{QC}. It is evidently not greater than $a$ and thus scale as $\mathcal{O}(1)$.

2. Perform the pure-state-based quantum state tomography on the state $\sigma$. Since $\sigma$ is generally low-rank, the time taken for tomography is $\mathcal{O}(\frac{M_1'M_1\log(M_1N)}{\epsilon_{eff}^3\epsilon^5})$, where we assume $\sigma$ (or equivalently $\overline{S}$) has $M_1'$ significantly large elements and $\epsilon_{eff}$ is the lower bound of eigenvalues that are taken into account.

3. Estimate the the supports of candidate 2-itemsets in $C_2$. In the pure-state-based quantum state tomography on $\sigma$, the pure state close to $|\psi_{\sigma}\rangle=\sum_{ij}\psi_{\sigma}^{ij}|i\rangle|j\rangle$ is created, in which $\psi_{\sigma}^{ij}=\frac{\sigma_{ij}}{\sqrt{\sum_{j}\gamma_j^2}}$ and $\gamma_j$ are eigenvalues of $\sigma$. Since $\sigma$ is generally approximately low-rank, the eigenvalues $\gamma_j$ can be estimated by quantum principal component analysis \cite{QPCA} in very little time. Then, the support of $\{I_{f_i},I_{f_j}\}$, $\overline{S}_{ij}=\frac{a_f\psi_{\sigma}^{ij}}{\sqrt{\sum_{j}\gamma_j^2}}$ are estimated.

4. Pick out the frequent candidate 2-itemsets to constitute the set $F_2$.

To obtain $F_2$, the Apriori algorithm directly computes the supports ($\overline{S}_{ij}$) of $\frac{M_1(M_1-1)}{2}$ candidate 2-itemsets in $C_2$ and thus takes runtime $\frac{M_1(M_1-1)N}{2}$.
To reduce the complexity, the sampling technique can also be applied and $\mathcal{O}(\frac{M_1^2}{\epsilon^2})$ samples for estimating each support are used to make sum of the squared error is within $\epsilon^2$. In our algorithm, $\mathcal{O}(\frac{M_1^{'}}{\epsilon^2})$ samples \cite{NS} in step 2 are required to make sum of the squared errors of estimating these supports is also within $\epsilon^2$ and thus total runtime scales as $\mathcal{O}(\frac{M_1'M_1\log(M_1N)}{\epsilon_{eff}^3\epsilon^5})$, which offers polynomial speedup over the classical sampling algorithm when $M_1'=\mathcal{O}(M_1^p)$ and $p<1$. In practice, since there will be a large number of pairs of items $\{I_i,I_j\}$ have weak associations, the speedup can be achieved with a high probability.

The comparison of time complexity between our quantum algorithm and classical algorithms for mining frequent 1-itemsets and 2-itemsets is given in the table \ref{tab:table1}. In is shown that our algorithm is exponentially faster than the Apriori algorithm due to the use of sampling technique and potentially polynomially faster than the classical sampling algorithm.

\begin{table}[b]
\caption{\label{tab:table1}%
Comparison of time complexity between our quantum algorithm and classical Apriori algorithm and sampling algorithm for mining frequent 1-itemsets $F_1$ and 2-itemsets $F_2$.}
\begin{ruledtabular}
\begin{tabular}{ccc}
\multirow{2}{*}{Algorithm} &\multicolumn{2}{c}{Time complexity} \\
\cline{2-3}
 &Mining $F_1$ &Mining $F_2$ \\ \hline
Apriori  &$\mathcal{O}(MN)$ &$\mathcal{O}(M_1^2N)$ \\
Sampling &$\mathcal{O}(\frac{M}{\epsilon^2})$ &$\mathcal{O}(\frac{M_1^2}{\epsilon^2})$ \\
Quantum   &$\mathcal{O}(\frac{\sqrt{M}\log(MN)}{\epsilon^2})$ &$\mathcal{O}(\frac{M_1'M_1\log(M_1N)}{\epsilon_{eff}^3\epsilon^5})$\\

\end{tabular}
\end{ruledtabular}
\end{table}
%
%


\emph{Conclusion.---}
In this letter, we have provided a quantum algorithm for ARM with focus on the most compute intensive process in ARM, mining frequent 1-itemsets and 2-itemsets. In our algorithm, the techniques of amplitude amplification and pure-state-based quantum state tomography are introduced to make our algorithm potentially achieves polynomial speedup over the classical sampling algorithm. Moreover, since limited quantum oracles accessing the database are used in our algorithm, data privacy for database can be achieved to some degree. We hope this algorithm can be useful for designing more quantum algorithms for big data mining tasks.

\section*{Acknowledgements}
This work is supported by NSFC (Grant Nos. 61272057, 61572081).


\begin{thebibliography}{ss3}
\bibitem{NS} In our algorithm, we estimate the supports of itemsets in $C_2$ by sampling the pure state $|\psi_{\sigma}\rangle=\sum_{ij}\psi_{\sigma}^{ij}|i\rangle|j\rangle$ and estimating $\psi_{\sigma}^{ij}$. Since $\overline{S}_{ij}=\frac{a_f\psi_{\sigma}^{ij}}{\sqrt{\sum_{j}\gamma_j^2}}$, the sum of the squared errors for estimating $\psi_{\sigma}^{ij}$ is set $\frac{\sum_j\gamma_j^2\epsilon^2}{a_f^2}$ to make sure the sum of squared error for estimating these supports is also within $\epsilon^2$. Moreover, $\sum_j\gamma_j^2 \in \mathcal{O}(1)$ due to the approximately low rank of $\sigma$. Hence, $\mathcal{O}(\frac{M_1^{'}a_f^2}{(\sum_j\gamma_j^2)\epsilon^2})=\mathcal{O}(\frac{M_1^{'}}{\epsilon^2})$ samples in step 2 are required.
\end{thebibliography}
\end{document}